\documentclass[11pt]{article}
\usepackage{amsfonts,amsmath,amsthm,amscd,amssymb,latexsym}
\usepackage[active]{srcltx}                    % SRC Specials: DVI [Inverse] Search

\theoremstyle{plain}
\newtheorem{thm}{Theorem}[section]

\newtheorem{lem}[thm]{Lemma}
\newtheorem{prop}[thm]{Proposition}

\theoremstyle{definition}
\theoremstyle{remark}
\numberwithin{equation}{section}
\newcommand{\keywords}{\textbf{Keywords: }\medskip}
\newcommand{\subjclass}{\textbf{Math. Subj. Clas.: }\medskip}
\title
 {A discrete  Dirac-K\"{a}hler  equation \\
using a geometric  discretisation  scheme}

\author{ Volodymyr Sushch \\  \\ \small Koszalin University of Technology, \\ \small Sniadeckich 2, 75-453 Koszalin, Poland \\
\small  e-mail: volodymyr.sushch@tu.koszalin.pl  }
\date{\small January 31, 2018}

\begin{document}

\maketitle

\begin{abstract}
Discrete models of the Dirac-K\"{a}hler equation and the Dirac equation in the Hestenes form are discussed.
A discrete version of the plane wave solutions to a discrete analogue of the Hestenes equation is established.
\end{abstract}



\section{Introduction}
We study a  discrete model of the  Dirac-K\"{a}hler  equation in which  key geometric aspects of the continuum counterpart are captured.
 We pay special attention to the description of some method of construction a discretization scheme based on the use of the differential forms calculus.
 The aim of this paper is to establish some discrete version of the plane wave solutions  to a discrete Dirac equation  in the Hestenes form.
 To construct these discrete solutions we introduce a Clifford product acting on the space of discrete inhomogeneous forms.
 This work is a direct continuation of  that described in my previous papers \cite{S1,S2, S3, S4}. In \cite{S4}, a correspondence between a discrete  Dirac-K\"{a}hler  equation and a discrete analogue of the Hestenes equation was studied.
 There are several approaches to study of discrete versions of the Dirac-K\"{a}hler equation based on the use  a discrete Clifford calculus framework on lattices. For a review of discrete Clifford analysis, we refer the reader to \cite{FKS,F2, F3, Kanamori, Vaz}.
 It is beyond the scope of this paper to fully discuss differences and intersections between approaches.

We first briefly review some definitions and basic facts  on the Dirac-K\"{a}hler equation \cite{Kahler, Rabin} and the Dirac equation
 in the spacetime algebra \cite{H1, H2}.
Let $M={\mathbb R}^{1,3}$ be  Minkowski space with  metric signature  $(+,-,-,-)$.
Denote by $\Lambda^r(M)$ the vector space of smooth differential $r$-forms, $r=0,1,2,3,4$. We consider  $\Lambda^r(M)$ over $\mathbb{C}$.
Let $\omega,\ \varphi\in\Lambda^r(M)$.
  The inner product is defined by
\begin{equation}\label{1.1}
(\omega, \ \varphi)=\int_{M}\omega\wedge\ast\overline{\varphi},
\end{equation}
where $\wedge$ is the exterior product and $\ast$ is the Hodge star operator  $\ast:\Lambda^r(M)\rightarrow\Lambda^{4-r}(M)$ with respect to the Lorentz metric.
Let $d:\Lambda^r(M)\rightarrow\Lambda^{r+1}(M)$ be the exterior differential and let $\delta:\Lambda^r(M)\rightarrow\Lambda^{r-1}(M)$ be the formal adjoint of $d$  with respect to  \eqref{1.1}. We have $$\delta=\ast d\ast.$$
 Denote by $\Lambda(M)$ the set of all differential forms on $M$. We have
\begin{equation*}
\Lambda(M)=\Lambda^{ev}(M)\oplus\Lambda^{od}(M),
\end{equation*}
where $\Lambda^{ev}(M)=\Lambda^0(M)\oplus\Lambda^2(M)\oplus\Lambda^4(M)$  and  $\Lambda^{od}(M)=\Lambda^1(M)\oplus\Lambda^3(M)$.
Let $\Omega\in\Lambda(M)$
be an inhomogeneous differential form, i.e.
$$\Omega=\sum_{r=0}^4\overset{r}{\omega},$$
where $\overset{r}{\omega}\in\Lambda^r(M)$.
 The Dirac-K\"{a}hler equation for a free electron is given by
\begin{equation}\label{1.2}
i(d+\delta)\Omega=m\Omega,
\end{equation}
where $i$ is the usual complex unit  and  $m$  is a mass parameter.
It is easy to show that Eq.~\eqref{1.2} is equivalent to the set of equations
\begin{eqnarray*}\label{}
i\delta\overset{1}{\omega}=m\overset{0}{\omega},\\
i(d\overset{0}{\omega}+\delta\overset{2}{\omega})=m\overset{1}{\omega},\\
i(d\overset{1}{\omega}+\delta\overset{3}{\omega})=m\overset{2}{\omega},\\
i(d\overset{2}{\omega}+\delta\overset{4}{\omega})=m\overset{3}{\omega},\\
id\overset{3}{\omega}=m\overset{4}{\omega}.
\end{eqnarray*}
Denote by $\Omega^{ev}$ and by $\Omega^{od}$  the even and odd parts of $\Omega$. It is clear that
\begin{align*}\label{}
i(d+\delta)\Omega^{od}=m\Omega^{ev}, \\
i(d+\delta)\Omega^{ev}=m\Omega^{od}.
\end{align*}
The operator $d+\delta$ is the analogue of the gradient operator in Minkowski spacetime
$$\nabla=\sum_{\mu=0}^3\gamma_\mu\partial^\mu, \quad \mu=0,1,2,3,$$
 where $\gamma_\mu$ is the Dirac gamma matrix.
 Think of $\{\gamma_0, \gamma_1, \gamma_2, \gamma_3\}$ as a vector basis in spacetime. Then the  gamma matrices $\gamma_\mu$ can be considered as generators of the Clifford algebra of spacetime $\mathcal{C}\ell(1,3)$ \cite{B, B1}. This algebra Hestenes \cite{H2} calls  the spacetime algebra.   Denote by $\emph{C}\ell_{\mathbb{R}}(1,3)$ $(\emph{C}\ell_{\mathbb{C}}(1,3))$ the real (complex) Clifford algebra. It is known that an inhomogeneous form $\Omega$ can be represented as element of $\mathcal{C}\ell_{\mathbb{C}}(1,3)$.  Then the Dirac-K\"{a}hler equation can be written as the algebraic equation
 \begin{equation}\label{1.3}
  i\nabla\Omega=m\Omega, \quad \Omega\in\emph{C}\ell_{\mathbb{C}}(1,3).
 \end{equation}
  Eq.~\eqref{1.3} is equivalent to the four Dirac equations (traditional column-spinor equations) for a free electron.  Let $\emph{C}\ell^{ev}(1,3)$ be the even subalgebra of the  algebra $\emph{C}\ell(1,3)$. The equation
 \begin{equation}\label{1.4}
-\nabla\Omega^{ev}\gamma_1\gamma_2=m\Omega^{ev}\gamma_0, \quad \Omega^{ev}\in\emph{C}\ell_{\mathbb{R}}^{ev}(1,3),
 \end{equation}
is called the Hestenes form of the Dirac equation \cite{H1, H2}.

It should be noted  that the grade algebra $\Lambda(M)$ endowed with the Clifford multiplication is an example of the the Clifford algebra.
 In this case the basis covectors $e^\mu=dx^\mu$ of spacetime are considered as generators of the Clifford algebra. Then the Hestenes  equation  \eqref{1.4} can be rewritten in terms of inhomogeneous forms as
 \begin{equation}\label{1.5}
 -(d+\delta)\Omega^{ev} e^1e^2=m\Omega^{ev} e^0, \quad \Omega^{ev}\in\Lambda_{\mathbb{R}}^{ev}(M).
 \end{equation}
\section{Discretization Scheme}
The starting point for consideration is a combinatorial model of Euclidean space. The proposed approach was originated by Dezin \cite{Dezin}.
For the convenience of the reader we briefly repeat the relevant material from \cite{S2}
without proofs, thus making our presentation self-contained.

Let the tensor product
\begin{equation*}
C(4)=C\otimes C\otimes C\otimes C
\end{equation*}
of a   1-dimensional complex be a combinatorial model of Euclidean space
 $\mathbb{R}^4$.
  The 1-dimensional complex $C$ is defined in the following way.
 Introduce  the sets  $\{x_\kappa\}$ and $\{e_{\kappa}\}$, $\kappa\in {\mathbb Z}$.
Let $C^0$ and $C^1$ be the free abelian groups of 0-dimensional and 1-dimensional
chains generated by $\{x_\kappa\}$ and $\{e_{\kappa}\}$.
The free abelian group is understood as the direct sum of infinity cyclic groups generated by $\{x_\kappa\}$, $\{e_{\kappa}\}$.
Arbitrary elements (chains) $a\in C^0$ and $b\in C^1$
can be written as the formal sums
$$
a=\sum_{\kappa}a^{\kappa}x_{\kappa}, \qquad  b=\sum_{\kappa}b^{\kappa}e_{\kappa}, \qquad a^{\kappa}, b^{\kappa}\in {\mathbb Z}.
$$
It is convenient to
introduce the shift operator  $\tau$ in the set of indices by
\begin{equation}\label{2.1}
\tau\kappa=\kappa+1.
\end{equation}
We define the boundary operator $\partial: C^1\rightarrow C^0, \quad  \partial: C^0\rightarrow 0,$
 setting
\begin{equation}\label{2.2}
\partial e_\kappa=x_{\tau\kappa}-x_\kappa, \qquad \partial x_\kappa=0.
\end{equation}
The definition \eqref{2.2} is extended to arbitrary chains by linearity.

The direct sum $C=C^0\oplus C^1$ with the boundary operator $\partial$
defines the 1-dimensional complex. It is known that a free abelian group is an abelian group with basis.   One can regard the sets $\{x_\kappa\}$,  $\{e_{\kappa}\}$ as  sets of basis elements of the groups  $C^0$ and $C^1$. Geometrically we can interpret the 0-dimensional basis elements $x_\kappa$ as points of the real line and the 1-dimensional basis elements $e_\kappa$ as open intervals between points. We call the complex  $C$ a combinatorial real line.

Multiplying the basis elements $x_\kappa$, $e_\kappa$ in various way we obtain
basis elements of $C(4)$.
Let $s_k^{(r)}$ be an arbitrary
$r$-dimensional basis element $C(4)$. Then we have
$$s_k^{(r)}=s_{k_0}\otimes s_{k_1}\otimes s_{k_2}\otimes s_{k_3},$$ where $s_{k_\mu}$
is either  $x_{k_\mu}$ or  $e_{k_\mu}$ and $k=(k_0,k_1,k_2,k_3)$,  $k_\mu\in \Bbb{Z}$.
The dimension $r$ of a basis element $s_k^{(r)}$ is given by the number of factors $e_{k_\mu}$ that appear in it.
The product  contains exactly $r$ of $1$-dimensional elements $e_{k_\mu}$ and $4-r$
of $0$-dimensional elements  $x_{k_\mu}$,  and  the superscript $(r)$ indicates also a position  of $e_{k_\mu}$ in $s_k^{(r)}$.
For example, the 1- and 2-dimensional basis elements
of $C(4)$ can be written as
\begin{align*}
e_k^0=e_{k_0}\otimes x_{k_1}\otimes x_{k_2}\otimes x_{k_3}, \nonumber \qquad
e_k^1=x_{k_0}\otimes e_{k_1}\otimes x_{k_2}\otimes x_{k_3}, \\
e_k^2=x_{k_0}\otimes x_{k_1}\otimes e_{k_2}\otimes x_{k_3},  \qquad
e_k^3=x_{k_0}\otimes x_{k_1}\otimes x_{k_2}\otimes e_{k_3} \
\end{align*}
and
\begin{align*}
e_k^{01}=e_{k_0}\otimes e_{k_1}\otimes x_{k_2}\otimes x_{k_3},  \qquad
e_k^{12}=x_{k_0}\otimes e_{k_1}\otimes e_{k_2}\otimes x_{k_3},\nonumber \\
e_k^{02}=e_{k_0}\otimes x_{k_1}\otimes e_{k_2}\otimes x_{k_3}, \qquad
e_k^{13}=x_{k_0}\otimes e_{k_1}\otimes x_{k_2}\otimes e_{k_3}, \nonumber \\
e_k^{03}=e_{k_0}\otimes x_{k_1}\otimes x_{k_2}\otimes e_{k_3}, \qquad
e_k^{23}=x_{k_0}\otimes x_{k_1}\otimes e_{k_2}\otimes e_{k_3}.
\end{align*}
Let $C(4)=C(p)\otimes C(q)$, where $p+q=4$ and $C(p)$ is the tensor product of $p$ factors of $C=C(1)$, $p=1,2,3$.
The definition \eqref{2.2} of  $\partial$  is  extended to  arbitrary chains of $C(4)$ by the rule
\begin{equation}\label{2.3}
\partial(a\otimes b)=\partial a\otimes b+(-1)^ra\otimes\partial b,
\end{equation}
where $a\in C(p)$ and $b\in C(q)$ and $r$ is the dimension of the chain $a$.
It is easy to check that  $\partial\partial a=0$ for any $a\in C(4)$.

Suppose that the combinatorial model of Minkowski space has the same structure
as $C(4)$.  We will use the index $k_0$ to denote the basis elements of $C$ which correspond to the time coordinate of
$M$. Hence, the indicated basis elements will be written as  $x_{k_0}$, $e_{k_0}$.

Let us introduce the construction of a double complex. This construction generalizes that of \cite{S2}. Together with
the complex $C(4)$ we consider its double, namely the complex
$\tilde{C}(4)$ of exactly the same structure.
The statement that $C(4)$ and $\tilde{C}(4)$ have the same structure means that the corresponding chains $a\in C(4)$ and $\tilde{a}\in \tilde{C}(4)$ have  the same coefficients.
Define the one-to-one
correspondence
\begin{equation*}
\ast^c : C(4)\rightarrow\tilde{C}(4), \qquad \ast^c : \tilde
C(4)\rightarrow C(4)
\end{equation*}
by the rule
\begin{equation}\label{2.4}
\ast^cs_k^{(r)}=\varepsilon(r)\tilde s_k^{(4-r)}, \qquad \ast^c\tilde s_k^{(r)}=\varepsilon(r) s_k^{(4-r)},
\end{equation}
where
\begin{equation*}
 \tilde s_k^{(4-r)}=\ast^c s_{k_0}\otimes \ast^c s_{k_1}\otimes
\ast^c s_{k_2}\otimes \ast^c s_{k_3}
\end{equation*}
and $\ast^c x_{k_\mu}=\tilde e_{k_\mu}, \
\ast^c e_{k_\mu}=\tilde x_{k_\mu}.$
Here $\varepsilon(r)$ is the Levi-Civita symbol, i.e.
\begin{equation*}
\varepsilon(r)=\left\{\begin{array}{l}+1 \quad  \mbox{if} \quad ((r),
(4-r))\quad \mbox{is an even permutation of} \quad (0,1,2,3) \\
                            -1 \quad  \mbox{if} \quad ((r),
(4-r))\quad \mbox{is an odd permutation of} \quad (0,1,2,3).
                            \end{array}\right.
\end{equation*}
For example, for the 1- and 2-dimensional basis elements  we have
\begin{equation*}\label{}
\ast^c e_k^0=\tilde e_k^{123},  \qquad \ast^c e_k^1=-\tilde e_k^{023}, \qquad
\ast^c e_k^2=\tilde e_k^{013}, \qquad \ast^c e_k^3=-\tilde e_k^{012}
\end{equation*}
and
\begin{align*}
\ast^c e_k^{01}=\tilde e_k^{23},  \qquad \ast^c e_k^{02}=-\tilde e_k^{13}, \qquad  \ast^c e_k^{03}=\tilde e_k^{12}, \\
\ast^c e_k^{12}=\tilde e_k^{03}, \qquad \ast^c e_k^{13}=-\tilde e_k^{02}, \qquad \ast^c e_k^{23}=\tilde e_k^{01}.
\end{align*}
The operation \eqref{2.4} is linearly extended to chains.
\begin{prop}Let $a_r\in C(4)$ be an $r$-dimensional chain,
then we have
\begin{equation*}\label{11}
\ast^c\ast^c
 a_r=(-1)^{r}a_r.
\end{equation*}
\end{prop}
 \begin{proof} The proof consists in applying the operation $\ast^c$ for basis elements.
\end{proof}
Let us now consider a dual complex to $C(4)$. We define it as the complex of cochains
$K(4)$ with complex coefficients. The complex $K(4)$
has a similar structure, namely ${K(4)=K\otimes K\otimes K\otimes K}$, where $K$ is a dual
complex to the 1-dimensional complex $C$. Let $x^\kappa$ and $e^\kappa$, $\kappa\in {\mathbb Z}$, be  the 0- and 1-dimensional basis elements of $K$. Then an arbitrary $r$-dimensional basis element of $K(4)$ can be written  as
$s_{(r)}^k=s^{k_0}\otimes s^{k_1}\otimes s^{k_2}\otimes s^{k_3}$, where $s^{k_\mu}$
is either  $x^{k_\mu}$ or  $e^{k_\mu}$ and $k=(k_0,k_1,k_2,k_3)$.   We will call cochains forms,
emphasizing their relationship with differential forms.
Denote by  $K^r(4)$ the set of all $r$-forms. Then $K(4)$ can be expressed by
\begin{equation*}
K(4)=K^{ev}(4)\oplus K^{od}(4),
\end{equation*}
 where $K^{ev}(4)=K^0(4)\oplus K^2(4)\oplus K^4(4)$
  and $K^{od}(4)=K^1(4)\oplus K^3(4).$
  The complex $K(4)$ is a discrete analogue of $\Lambda(M)$.
  Let $\overset{r}{\omega}\in K^r(4)$, then we have
\begin{equation*}
\overset{0}{\omega}=\sum_k\overset{0}{\omega}_kx^k,  \quad  \overset{2}{\omega}=\sum_k\sum_{\mu<\nu} \omega_k^{\mu\nu}e_{\mu\nu}^k, \quad \overset{4}{\omega}=\sum_k\overset{4}{\omega}_ke^k,
\end{equation*}
\begin{equation*}
\overset{1}{\omega}=\sum_k\sum_{\mu=0}^3\omega_k^\mu e_\mu^k, \quad
\overset{3}{\omega}=\sum_k\sum_{\iota<\mu<\nu} \omega_k^{\iota\mu\nu}e_{\iota\mu\nu}^k,
\end{equation*}
where $e_\mu^k$, $e_{\mu\nu}^k$ and  $e_{\iota\mu\nu}^k$ are the 1-, 2- and 3-dimensional basis elements of $K(4)$, and $x^k=x^{k_0}\otimes x^{k_1}\otimes x^{k_2}\otimes x^{k_3}$, \  $e^k=e^{k_0}\otimes e^{k_1}\otimes e^{k_2}\otimes e^{k_3}$.
The components $\overset{0}{\omega}_k, \ \overset{4}{\omega}_k, \  \omega_k^\mu, \ \omega_k^{\mu\nu}$ and $\omega_k^{\iota\mu\nu}$ are complex numbers.

We define the pairing (chain-cochain) operation for any basis elements
$\epsilon_k\in C(4)$,  $s^k\in K(4)$ by the rule
\begin{equation}\label{2.5}
\langle\epsilon_k, \ s^k\rangle=\left\{\begin{array}{l}0, \quad \epsilon_k\ne s_k\\
                            1, \quad \epsilon_k=s_k.
                            \end{array}\right.
\end{equation}
The operation \eqref{2.5}  is linearly extended to arbitrary chains-cochains.

The coboundary operator $d^c: K^r(4)\rightarrow K^{r+1}(4)$ is defined by
\begin{equation}\label{2.6}
\langle\partial a_{r+1}, \ \overset{r}{\omega}\rangle=\langle a_{r+1}, \ d^c\overset{r}{\omega}\rangle,
\end{equation}
where $a_{r+1}\in C(4)$ is an $r+1$ dimensional chain. The operator $d^c$ is an analog of the exterior differential.
From the above it follows that
\begin{equation*}\label{}
 d^c\overset{4}{\omega}=0 \quad \mbox{and} \quad d^cd^c\overset{r}{\omega}=0 \quad \mbox{for any} \quad r.
\end{equation*}
Let the difference operator $\Delta_\mu$ be defined by
\begin{equation}\label{2.7}
\Delta_\mu\omega_k^{(r)}=\omega_{\tau_\mu k}^{(r)}-\omega_k^{(r)},
\end{equation}
where  $\omega_k^{(r)}\in\mathbb{C}$ is a component of $\overset{r}{\omega}\in K^r(4)$ and
$\tau_\mu$ is   the shift operator  which acts  as
\begin{equation*}\label{}
\tau_\mu k=(k_0,...k_\mu+1,...k_3), \quad   \mu=0,1,2,3,
  \end{equation*}
  where $\tau$ is defined by \eqref{2.1}.
  Using \eqref{2.3} and \eqref{2.6} we can calculate
\begin{equation}\label{2.8}
d^c\overset{0}{\omega}=\sum_k\sum_{\mu=0}^3(\Delta_\mu\overset{0}{\omega}_k)e_\mu^k, \qquad
d^c\overset{1}{\omega}=\sum_k\sum_{\mu<\nu}(\Delta_\mu\omega_k^\nu-\Delta_\nu\omega_k^\mu)e_{\mu\nu}^k,
\end{equation}
\begin{align}\label{2.9}
d^c\overset{2}{\omega}=\sum_k\big[(\Delta_0\omega_k^{12}-\Delta_1\omega_k^{02}+\Delta_2\omega_k^{01})e_{012}^k+\nonumber \\
+(\Delta_0\omega_k^{13}-\Delta_1\omega_k^{03}+\Delta_3\omega_k^{01})e_{013}^k+ \nonumber \\
+(\Delta_0\omega_k^{23}-\Delta_2\omega_k^{03}+\Delta_3\omega_k^{02})e_{023}^k+ \nonumber \\
+(\Delta_1\omega_k^{23}-\Delta_2\omega_k^{13}+\Delta_3\omega_k^{12})e_{123}^k\big],
\end{align}
\begin{equation}\label{2.10}
d^c\overset{3}{\omega}=\sum_k(\Delta_0\omega_k^{123}-\Delta_1\omega_k^{023}+\Delta_2\omega_k^{013}-\Delta_3\omega_k^{012})e^k.
\end{equation}

Let  $\tilde K(4)$ be a complex of the cochains over the double complex
$\tilde C(4)$, with the coboundary operator $d^c$ defined in it by
\eqref{2.6}. Hence, $\tilde K(4)$ has the same structure as  $K(4)$.
This means that the corresponding forms $\omega\in K(4)$
 and $\tilde\omega\in \tilde K(4)$ have  the same components.

Let us introduce a discrete version of the Hodge star operator by using the double complex construction.
Define the operation $\ast: K^r(4)\rightarrow \tilde K^{4-r}(4)$ for the basis element $s^k_{(r)}\in K^r(4)$ by the rule
\begin{equation}\label{2.11}
\ast s^k_{(r)}=Q(k_0)\varepsilon(r)\tilde s^k_{(4-r)},
\end{equation}
where
\begin{equation*}
Q(k_0)=\left\{\begin{array}{l}+1 \quad  \mbox{if} \quad s^{k_0}=x^{k_0} \\
                            -1 \quad  \mbox{if} \quad s^{k_0}= e^{k_0}.
                            \end{array}\right.
\end{equation*}
This definition makes sense because the formula \eqref{2.11} preserves the  Lorentz  signature of metric on $K(4)$.
More explicitly,  we have
\begin{equation*}\label{21}
\ast x^k=\tilde e^k, \qquad \ast e^k=-\tilde x^k,
\end{equation*}
\begin{equation*}\label{23}
\ast e_0^k=-\tilde e_{123}^k, \qquad \ast e_1^k=-\tilde e_{023}^k, \qquad
\ast e_2^k=\tilde e_{013}^k, \qquad \ast e_3^k=-\tilde e_{012}^k,
\end{equation*}
\begin{align*}\label{24}
\ast e_{01}^k&=-\tilde e_{23}^ k, \qquad \ast e_{02}^k=\tilde e_{13}^ k, \qquad \ast e_{03}^k=-\tilde e_{12}^ k, \nonumber \\
\ast e_{12}^k&=\tilde e_{03}^ k, \qquad  \ast e_{13}^k=-\tilde e_{02}^ k, \qquad \ast e_{23}^k=\tilde e_{01}^ k,
\end{align*}
\begin{equation*}\label{25}
\ast e_{012}^k=-\tilde e_3^ k, \quad \ast e_{013}^k=\tilde e_2^ k, \quad
\ast e_{023}^k=-\tilde e_1^ k, \quad \ast e_{123}^k=-\tilde e_0^k.
\end{equation*}
For any $r$-form the operation  \eqref{2.11} is extended by linearity. Likewise, the mapping
$\ast: \tilde K^r(4)\rightarrow  K^{4-r}(4)$ is given by the rule  \eqref{2.11}.
\begin{prop} Let $\overset{r}{\omega}\in K^r(4)$,
then we have
\begin{equation}\label{2.12}
\ast\ast\overset{r}{\omega}=
(-1)^{r+1}\overset{r}{\omega}.
 \end{equation}
\end{prop}
 \begin{proof} The operation  $\ast$ is linear. It is easy to check that by definition, the composition of $\ast$ with itself gives
 \begin{equation*}
 \ast\ast s^k_{(r)}=(-1)^{r(4-r)+1}s^k_{(r)}=(-1)^{r+1}s^k_{(r)}
 \end{equation*}
 for any basis element $s^k_{(r)}\in K^r(4)$.
\end{proof}
Hence, the operation $\ast^2$ is either an involution or antiinvolution. It means that the discrete $\ast$ operation imitates correctly the continuum case.
\begin{prop}
Let $\tilde a_r\in \tilde C(4)$ be an $r$-dimensional chain and let $\omega\in K^{4-r}(4)$.
Then we have
\begin{equation*}\label{27}
\langle \tilde a_r, \ \ast\omega\rangle=(-1)^rQ(k_0)\langle \ast^c \tilde a_r, \ \omega\rangle.
\end{equation*}
 \end{prop}
\begin{proof}
See Proposition~3 in \cite{S2}.
\end{proof}
Let us consider the  $4$-dimensional finite chain  $e_n\subset C(4)$ of the form:
\begin{equation}\label{2.13}
e_n=\sum_ke_k, \quad k=(k_0,k_1,k_2,k_3),\quad k_\mu=1,2, ...,n_\mu,
\end{equation}
where $n_\mu\in \mathbb{N}$ is a fixed number for each $\mu=0,1,2,3$.
This finite sum of $4$-dimensional basis elements of  $C(4)$ imitates a domain of $M$.
We set
\begin{equation*}\label{}
 V_r=\sum_k\sum_{(r)}s_k^{(r)}\otimes\ast^c s_k^{(r)}, \quad k_\mu=1,2, ...,n_\mu.
\end{equation*}
For example,
\begin{equation*}
 V_1=\sum_k\sum_{\mu=0}^3 e_k^{\mu}\otimes\ast^c e_k^{\mu}=\sum_k\big(e_k^0\otimes\tilde e_k^{123}-
 e_k^1\otimes\tilde e_k^{023}+e_k^2\otimes\tilde e_k^{013}-e_k^3\otimes\tilde
 e_k^{012}\big).
 \end{equation*}
  Let
 \begin{equation*}
 \mathbb{V}=\sum_{r=0}^4 V_r.
 \end{equation*}
For any $r$-forms $\varphi, \omega\in K^r(4)$  the inner
 product over the set  $e_n$ \eqref{2.13} is defined  by the rule
 \begin{equation}\label{2.14}
 (\varphi ,\ \omega)_{e_n}=\langle\mathbb{V}, \ \varphi\otimes\ast\overline{\omega}\rangle,
 \end{equation}
  where $\overline{\omega}$ denotes the complex conjugate of the form $\omega$. For  forms of different degrees the product \eqref{2.14} is set equal to zero.
  It is clear that
  \begin{equation*}
   (\varphi ,\ \omega)_{e_n}=\langle V_r, \ \varphi\otimes\ast\overline{\omega}\rangle
 =\sum_k\sum_{(r)}
 \langle s_k^{(r)},\
\varphi \rangle\langle\ast^c s_k^{(r)}, \ \ast\overline{\omega}\rangle.
 \end{equation*}
 For example, if $\overset{1}{\varphi}, \ \overset{1}{\omega}\in K^1(4)$ then we obtain
\begin{equation*}\label{}
(\overset{1}{\varphi}, \ \overset{1}{\omega})_{e_n}=\sum_k\big(-\varphi_k^0\overline{\omega}_k^0+\varphi_k^1\overline{\omega}_k^1+\varphi_k^2\overline{\omega}_k^2+
\varphi_k^3\overline{\omega}_k^3\big).
\end{equation*}
It should be noted that the definition of the inner product correctly imitates the continual case \eqref{1.1} and the  Lorentz  metric structure is still  captured here.

The inner product \eqref{2.14} makes it possible to define the adjoint of
$d^c$, denoted $\delta^c$.
\begin{prop} For any $(r-1)$-form $\varphi$ and $r$-form $\omega$ we have
\begin{equation}\label{2.15}
(d^c\varphi, \ \omega)_{e_n}=\langle\partial\mathbb{V}, \ \varphi\otimes\ast\overline{\omega}\rangle+( \varphi, \ \delta^c\omega)_{e_n},
\end{equation} where
 \begin{equation}\label{2.16}
 \delta^c=(-1)^{r}\ast^{-1}d^c\ast
\end{equation}
and $\ast^{-1}$ is the inverse operation of $\ast$.
\end{prop}
\begin{proof}
The proof is a computation. See Proposition~4 in \cite{S2}.
\end{proof}
The relation \eqref{2.15} is a discrete analog of the Green formula.
From \eqref{2.12} we infer
\begin{equation*}
\ast^{-1}=(-1)^{r(4-r)+1}\ast=(-1)^{r+1}\ast.
\end{equation*}
Putting this in \eqref{2.16} we obtain
 \begin{equation}\label{2.17}
 \delta^c=\ast d^c\ast.
\end{equation}
This makes it clear that the operator $\delta^c: K^{r+1}(4) \rightarrow K^r(4)$ is a discrete analog of the codifferential $\delta$.
 From \eqref{2.17} it follows that
 \begin{equation*}\label{}
 \delta^c\overset{0}{\omega}=0 \quad \mbox{and} \quad \delta^c\delta^c\overset{r}{\omega}=0 \quad \mbox{for any} \quad r.
\end{equation*}
  Using the definitions of $d^c$  and $\ast$ we can calculate
\begin{equation}\label{2.18}
\delta^c\overset{1}{\omega}=\sum_k(\Delta_0\omega_k^{0}-\Delta_1\omega_k^{1}-\Delta_2\omega_k^{2}-\Delta_3\omega_k^{3})x^k,
\end{equation}
\begin{align}\label{2.19}
\delta^c\overset{2}{\omega}=\sum_k\big[(\Delta_1\omega_k^{01}+\Delta_2\omega_k^{02}+\Delta_3\omega_k^{03})e_{0}^k \nonumber \\
+(\Delta_0\omega_k^{01}+\Delta_2\omega_k^{12}+\Delta_3\omega_k^{13})e_{1}^k \nonumber \\
+(\Delta_0\omega_k^{02}-\Delta_1\omega_k^{12}+\Delta_3\omega_k^{23})e_{2}^k \nonumber \\
+(\Delta_0\omega_k^{03}-\Delta_1\omega_k^{13}-\Delta_2\omega_k^{23})e_{3}^k\big],
\end{align}
\begin{align}\label{2.20}
\delta^c\overset{3}{\omega}=\sum_k\big[(-\Delta_2\omega_k^{012}-\Delta_3\omega_k^{013})e_{01}^k+
(\Delta_1\omega_k^{012}-\Delta_3\omega_k^{023})e_{02}^k \nonumber \\
+(\Delta_1\omega_k^{013}+\Delta_2\omega_k^{023})e_{03}^k
+(\Delta_0\omega_k^{012}-\Delta_3\omega_k^{123})e_{12}^k \nonumber \\
+(\Delta_0\omega_k^{013}+\Delta_2\omega_k^{123})e_{13}^k
+(\Delta_0\omega_k^{023}-\Delta_1\omega_k^{123})e_{23}^k\big],
\end{align}
\begin{align}\label{2.21}
\delta^c\overset{4}{\omega}=\sum_k\big[(\Delta_3\omega_k^{4})e_{012}^k-(\Delta_2\omega_k^{4})e_{013}^k
+(\Delta_1\omega_k^{4})e_{023}^k+(\Delta_0\omega_k^{4})e_{123}^k\big].
\end{align}
The linear map
\begin{equation*}
\Delta^c=-(d^c\delta^c+\delta^cd^c): \ K^r(4) \rightarrow K^r(4)
\end{equation*}
is called a discrete analogue of the Laplacian. It is clear that
\begin{equation*}\label{37}
-(d^c\delta^c+\delta^cd^c)=(d^c-\delta^c)^2=-(d^c+\delta^c)^2.
\end{equation*}
\section{Discrete Dirac-K\"{a}hler and Hestenes Equations}
Let us introduce a discrete inhomogeneous form as follows
\begin{equation*}\label{40}
\Omega=\sum_{r=0}^4\overset{r}{\omega},
\end{equation*}
where $\overset{r}{\omega}\in K^r(4)$.
A discrete analog of the Dirac-K\"{a}hler equation \eqref{1.2} is defined to be
\begin{equation}
i(d^c+\delta^c)\Omega=m\Omega,
\end{equation}
where $i$ is the usual complex unit and $m$ is a positive number (mass parameter).
We can write this equation more explicitly by separating its homogeneous components as
\begin{align*}\label{}
i\delta^c\overset{1}{\omega}=m\overset{0}{\omega},\\ \nonumber
i(d^c\overset{0}{\omega}+\delta^c\overset{2}{\omega})=m\overset{1}{\omega},\\
i(d^c\overset{1}{\omega}+\delta^c\overset{3}{\omega})=m\overset{2}{\omega},\\ \nonumber
i(d^c\overset{2}{\omega}+\delta^c\overset{4}{\omega})=m\overset{3}{\omega},\\ \nonumber
id^c\overset{3}{\omega}=m\overset{4}{\omega}.
\end{align*}
Moreover, by \eqref{2.8}--\eqref{2.10}  and \eqref{2.18}--\eqref{2.21}, the system of equations above for each $k=(k_0, k_1, k_2, k_3)$ expresses in  terms of 16 difference equations as follows

\begin{align*}\label{}
i(\Delta_0\omega_k^{0}-\Delta_1\omega_k^{1}-\Delta_2\omega_k^{2}-\Delta_3\omega_k^{3})=m\overset{0}{\omega}_k,\\
i(\Delta_0\overset{0}{\omega}_k+\Delta_1\omega_k^{01}+\Delta_2\omega_k^{02}+\Delta_3\omega_k^{03})=m\omega_k^0,\\
i(\Delta_1\overset{0}{\omega}_k+\Delta_0\omega_k^{01}+\Delta_2\omega_k^{12}+\Delta_3\omega_k^{13})=m\omega_k^1,\\
i(\Delta_2\overset{0}{\omega}_k+\Delta_0\omega_k^{02}-\Delta_1\omega_k^{12}+\Delta_3\omega_k^{23})=m\omega_k^2,\\
i(\Delta_3\overset{0}{\omega}_k+\Delta_0\omega_k^{03}-\Delta_1\omega_k^{13}-\Delta_2\omega_k^{23})=m\omega_k^3,\\
i(\Delta_0\omega_k^1-\Delta_1\omega_k^0-\Delta_2\omega_k^{012}-\Delta_3\omega_k^{013})=m\omega_k^{01},\\
i(\Delta_0\omega_k^2-\Delta_2\omega_k^0+\Delta_1\omega_k^{012}-\Delta_3\omega_k^{023})=m\omega_k^{02},\\
i(\Delta_0\omega_k^3-\Delta_3\omega_k^0+\Delta_1\omega_k^{013}+\Delta_2\omega_k^{023})=m\omega_k^{03},\\
i(\Delta_1\omega_k^2-\Delta_2\omega_k^1+\Delta_0\omega_k^{012}-\Delta_3\omega_k^{123})=m\omega_k^{12},\\
i(\Delta_1\omega_k^3-\Delta_3\omega_k^1+\Delta_0\omega_k^{013}+\Delta_2\omega_k^{123})=m\omega_k^{13},\\
i(\Delta_2\omega_k^3-\Delta_3\omega_k^2+\Delta_0\omega_k^{023}-\Delta_1\omega_k^{123})=m\omega_k^{23}, \\
i(\Delta_0\omega_k^{12}-\Delta_1\omega_k^{02}+\Delta_2\omega_k^{01}+\Delta_3\omega_k^{4})=m\omega_k^{012},\\
i(\Delta_0\omega_k^{13}-\Delta_1\omega_k^{03}+\Delta_3\omega_k^{01}-\Delta_2\omega_k^{4})=m\omega_k^{013},\\
i(\Delta_0\omega_k^{23}-\Delta_2\omega_k^{03}+\Delta_3\omega_k^{02}+\Delta_1\omega_k^{4})=m\omega_k^{023},\\
i(\Delta_1\omega_k^{23}-\Delta_2\omega_k^{13}+\Delta_3\omega_k^{12}+\Delta_0\omega_k^{4})=m\omega_k^{123}, \\
i(\Delta_0\omega_k^{123}-\Delta_1\omega_k^{023}+\Delta_2\omega_k^{013}-\Delta_3\omega_k^{012})=m\overset{4}{\omega}_k.
\end{align*}

Let us define  the Clifford multiplication in  $K(4)$ by the following rules:
\begin{align*}
&\mbox{(a)} \quad x^kx^k=x^k, \quad x^ke^k_\mu=e^k_\mu x^k=e^k_\mu,\\
&\mbox{(b)} \quad e^k_\mu e^k_\nu+e^k_\nu e^k_\mu=2g_{\mu\nu}x^k, \quad g_{\mu\nu}=\mbox{diag}(1,-1,-1,-1),\\
&\mbox{(c)} \quad e^k_{\mu_1}\cdots e^k_{\mu_s}=e^k_{\mu_1\cdots \mu_s} \quad \mbox{for} \quad 0\leq \mu_1<\cdots <\mu_s\leq 3,
\end{align*}
supposing the product to be zero in all other cases.

The operation is linearly extended to arbitrary discrete forms. For example, for   $\overset{1}{\omega},  \overset{1}{\varphi}\in K^1(4)$ we have
\begin{align*}
\overset{1}{\omega}\overset{1}{\varphi}=\big(\sum_k\sum_{\mu=0}^3\omega^\mu_ke_\mu^k\big)\big(\sum_k\sum_{\mu=0}^3\varphi^\mu_ke_\mu^k\big)=
\sum_k (\omega^0_k\varphi^0_k-\omega^1_k\varphi^1_k-\omega^2_k\varphi^2_k-\omega^3_k\varphi^3_k)x^k \\
+\sum_k [(\omega^0_k\varphi^1_k-\omega^1_k\varphi^0_k)e_{01}^k+(\omega^0_k\varphi^2_k-\omega^2_k\varphi^0_k)e_{02}^k+
(\omega^0_k\varphi^3_k-\omega^3_k\varphi^0_k)e_{03}^k\\
+(\omega^1_k\varphi^2_k-\omega^2_k\varphi^1_k)e_{12}^k+(\omega^1_k\varphi^3_k-\omega^3_k\varphi^1_k)e_{13}^k+(\omega^2_k\varphi^3_k-\omega^3_k\varphi^2_k)e_{23}^k].
\end{align*}
Consider the following unit forms
\begin{equation}\label{3.2}
x=\sum_kx^k, \quad e=\sum_ke^k, \quad e_\mu=\sum_ke_\mu^k, \quad e_{\mu\nu}=\sum_ke_{\mu\nu}^k.
\end{equation}
Note that the unit 0-form $x$ plays  a role of the unit element in $K(4)$, i.e. for any $r$-form  $\overset{r}{\omega}$ we have
\begin{equation*}
x\overset{r}{\omega}=\overset{r}{\omega}x=\overset{r}{\omega}.
\end{equation*}
\begin{prop}
The following holds:
\begin{equation*}
e_\mu e_\nu+e_\nu e_\mu=2g_{\mu\nu}x, \qquad \mu,\nu=0,1,2,3.
\end{equation*}
\end{prop}
\begin{proof}
By the rule (b) it is obvious.
\end{proof}
\begin{prop}  For any inhomogeneous form $\Omega\in K(4)$ we have
\begin{equation}\label{3.3}
(d^c+\delta^c)\Omega=\sum_{\mu=0}^3e_\mu\Delta_\mu\Omega,
\end{equation}
where
 $\Delta_\mu$ is the difference operator which acts on each component of $\Omega$ by the rule \eqref{2.7}.
\end{prop}
\begin{proof}
See Proposition~1 in \cite{S4}.
\end{proof}
Clearly,  the discrete Dirac-K\"{a}hler equation can be rewritten in the form
\begin{equation*}
i\sum_{\mu=0}^3e_\mu\Delta_\mu\Omega=m\Omega.
\end{equation*}
Let $\Omega^{ev}\in K^{ev}(4)$ be a real-valued even inhomogeneous form, i.e. $\Omega^{ev}=\overset{0}{\omega}+\overset{2}{\omega}+\overset{4}{\omega}$.
A discrete analogue of the Hestenes equation \eqref{1.5} is defined by
\begin{equation}\label{3.4}
-(d^c+\delta^c)\Omega^{ev} e_1e_2=m\Omega^{ev} e_0,
\end{equation}
where $e_0, e_1, e_2$ are given by  \eqref{3.2}. From \eqref{3.3} it follows that Eq.~\eqref{3.4} is equivalent to
\begin{equation*}
-\sum_{\mu=0}^3\big(e_\mu\Delta_\mu\Omega^{ev}\big)e_1e_2=m\Omega^{ev} e_0.
\end{equation*}
This equation  can be expressed in terms of difference equations as
\begin{eqnarray*}\label{}
\Delta_0\omega_k^{12}-\Delta_1\omega_k^{02}+\Delta_2\omega_k^{01}+\Delta_3\overset{4}{\omega}_k=m\overset{0}{\omega}_k,\\
\Delta_2\overset{0}{\omega}_k+\Delta_0\omega_k^{02}-\Delta_1\omega_k^{12}+\Delta_3\omega_k^{23}=m\omega_k^{01},\\
-\Delta_1\overset{0}{\omega}_k-\Delta_0\omega_k^{01}-\Delta_2\omega_k^{12}-\Delta_3\omega_k^{13}=m\omega_k^{02},\\
-\Delta_1\omega_k^{23}+\Delta_2\omega_k^{13}-\Delta_3\omega_k^{12}-\Delta_0\overset{4}{\omega}_k=m\omega_k^{03},\\
-\Delta_0\overset{0}{\omega}_k-\Delta_1\omega_k^{01}-\Delta_2\omega_k^{02}-\Delta_3\omega_k^{03}=m\omega_k^{12},\\
-\Delta_0\omega_k^{23}+\Delta_2\omega_k^{03}-\Delta_3\omega_k^{02}-\Delta_1\overset{4}{\omega}_k=m\omega_k^{13},\\
\Delta_0\omega_k^{13}-\Delta_1\omega_k^{03}+\Delta_3\omega_k^{01}-\Delta_2\overset{4}{\omega}_k=m\omega_k^{23},\\
\Delta_3\overset{0}{\omega}_k+\Delta_0\omega_k^{03}-\Delta_1\omega_k^{13}-\Delta_2\omega_k^{23}=m\overset{4}{\omega}_k.
\end{eqnarray*}

In \cite[Proposition 5]{S4},  it is proven  that by a solution of the discrete Dirac-K\"{a}hler equation \eqref{3.3} four independent real-valued solutions of the discrete Hestenes equation \eqref{3.4}
are constructed. This is a discrete version of the well-known result for
corresponding continuum equations \cite{B1}.

\section{Plane Wave Solutions}
Let us consider the real-valued forms
\begin{equation}\label{4.1}
 \Psi^\pm=\sum_k\Psi^\pm_kx^k,
 \end{equation}
 where
 \begin{equation}\label{4.2}
  \Psi^\pm_k=(x\pm p_0e_{12})^{k_0}(x\pm p_1e_{12})^{k_1}(x\pm p_2e_{12})^{k_2}(x\pm p_3e_{12})^{k_3},
 \end{equation}
 and $p_\mu\in\mathbb{R}$.
 Recall that $e_{12}$ is the unit 2-form given by \eqref{3.2}.
 It is easy to check that
\begin{equation}\label{4.3}
  \Delta_\mu\Psi_k^\pm=\pm p_\mu\Psi_k^\pm e_{12}, \quad \mu=0,1,2,3.
   \end{equation}
   Therefore
\begin{equation}\label{4.4}
   (d^c+\delta^c)\Psi^\pm=\sum_{\mu=0}^3e_\mu\Delta_\mu \Psi^\pm=\pm\sum_{\mu=0}^3e_\mu p_\mu \Psi^\pm e_{12}.
 \end{equation}
It should be noted  that the components $\Psi^\pm_k$ can be represented as
 \begin{equation*}
  \Psi^\pm_k=\sum_k(\psi^{\pm}_kx+\phi^{\pm}_ke_{12}),
 \end{equation*}
 where
 $$\psi^{\pm}_k=\psi^{\pm}_k(p_0,p_1,p_2,p_3) \quad \mbox{and}  \quad  \phi^{\pm}_k=\phi^{\pm}_k(p_0,p_1,p_2,p_3).$$
 Hence $\Psi^\pm$ are inhomogeneous real-valued  even forms of the form
 \begin{equation*}
 \Psi^\pm=\psi^{\pm}+\phi^{\pm},
 \end{equation*}
  where
  \begin{equation*}
  \psi^\pm=\sum_k\psi^{\pm}_kx^k, \qquad \phi^{\pm}=\sum_k\phi^{\pm}_ke_{12}^k.
  \end{equation*}
  We wish to find a solution of the discrete Hestenes equation \eqref{3.4} of the form
\begin{equation}\label{4.5}
 \Omega^\pm=A\Psi^\pm,
 \end{equation}
 where
   $A\in K^{ev}(4)$ is a constant real-valued form.
Hence $A$ can be expanded as
 \begin{equation}\label{4.6}
  A=\alpha^0x+\sum_{\mu<\nu}\alpha^{\mu\nu}e_{\mu\nu}+\alpha^4e,
 \end{equation}
   where $\alpha^0, \alpha^{\mu\nu}, \alpha^4\in\mathbb{R}$  and $x$, $e_{\mu\nu}$, $e$ are the unit forms  given by  \eqref{3.2}.
 We have
 \begin{eqnarray*}
(d^c+\delta^c)\Omega^\pm=(d^c+\delta^c)A\Psi^\pm=\sum_{\mu=0}^3e_\mu\Delta_\mu (A\Psi^\pm)\\
=\sum_{\mu=0}^3e_\mu A(\Delta_\mu\Psi^\pm)=
\sum_{\mu=0}^3e_\mu A\sum_k(\Delta_\mu\Psi^\pm_k) x^k\\
=\pm\sum_{\mu=0}^3e_\mu p_\mu A\sum_k\Psi^\pm_k x^k e_{12}=\pm\sum_{\mu=0}^3e_\mu p_\mu A\Psi^\pm e_{12}.
\end{eqnarray*}
Substituting this into the discrete Hestenes equation \eqref{3.4} we obtain
 \begin{equation*}
  \mp\big(\sum_{\mu=0}^3e_\mu p_\mu A\Psi^\pm e_{12}\big)e_{12}=mA\Psi^\pm e_0.
 \end{equation*}
By definition, $e_{12}e_{12}=-x$. This yields
 \begin{equation*}
  \pm\sum_{\mu=0}^3e_\mu p_\mu A\Psi^\pm=mA\Psi^\pm e_0.
 \end{equation*}
 Since $\Psi^\pm e_0 = e_0\Psi^\pm$, the equation above reduces to
 \begin{equation}\label{4.7}
  \sum_{\mu=0}^3e_\mu p_\mu A=mAe_0
 \end{equation}
  in the case of the form $\Psi^+$ and to
\begin{equation}\label{4.8}
  -\sum_{\mu=0}^3e_\mu p_\mu A=mAe_0
 \end{equation}
 in the case of the form $\Psi^-$. Firstly, take the form $\Psi^-$.
 Eq.~\eqref{4.8} implies
 \begin{equation}\label{4.9}
   -(p_0x+p_1e_0e_1+p_2e_0e_2+p_3e_0e_3)A=me_0A e_0,
 \end{equation}
 since $e_0e_0=x$.
 By trivial computation one finds that
 \begin{equation*}
   \big(p_0x-\sum_{\mu=1}^3p_\mu e_0 e_\mu\big)\big(p_0x+\sum_{\mu=1}^3p_\mu e_0 e_\mu\big)=\big(p_0^2-\sum_{\mu=1}^3p_\mu^2\big)x.
 \end{equation*}
 Then multiplying both sides of Eq.~\eqref{4.9} by $-(p_0x-p_1e_0e_1-p_2e_0e_2-p_3e_0e_3)$  we obtain

\begin{equation*}
  \big(p_0^2-\sum_{\mu=1}^3p_\mu^2\big)xA=-m\big(p_0x-\sum_{\mu=1}^3p_\mu e_0 e_\mu\big)e_0A e_0,
   \end{equation*}
   or equivalently,
  \begin{equation*}
  \big(p_0^2-\sum_{\mu=1}^3p_\mu^2\big)A=-m\big(\sum_{\mu=0}^3e_\mu p_\mu A\big) e_0.
   \end{equation*}
 Applying  \eqref{4.8} gives
  \begin{equation*}
   \big(p_0^2-\sum_{\mu=1}^3p_\mu^2\big)A=m^2Ae_0e_0,
 \end{equation*}
or equivalently,
 \begin{equation*}
   \big(p_0^2-\sum_{\mu=1}^3p_\mu^2\big)A=m^2A.
 \end{equation*}
 Thus Eq.~\eqref{4.8} has a non-trivial solution if and only if
 \begin{equation}\label{4.10}
   p_0^2-p_1^2-p_2^2-p_3^2=m^2
 \end{equation}
 or
 \begin{equation*}
   p_0=\pm\sqrt{m^2+p_1^2+p_2^2+p_3^2}.
 \end{equation*}
 \begin{prop}
  The form $A\Psi^-$  is a non-trivial solution of the discrete Hestenes equation \eqref{3.4} if and only the condition \eqref{4.10}
  holds.
   \end{prop}
  It is clear that the  same is true for  $A\Psi^+$ in place $A\Psi^-$, i.e. if we take Eq.~\eqref{4.7} we obtain the condition \eqref{4.10} again.

  Let $p=\{ p_0, p_1, p_2, p_3\}$ be the energy-momentum vector of a particle with (proper) mass $m$. Then the relation \eqref{4.10} is the energy-momentum relation.
  It is known that in the continuum theory the Hestenes equation \eqref{1.4} admits the plane wave solutions of the form
  \begin{equation*}
  \Phi^\pm=A\exp(\pm \gamma_2\gamma_1p \cdot x).
 \end{equation*}
 Thus the forms $\Omega^\pm=A\Psi^\pm$ are discrete versions of the plane-wave solutions $\Phi^\pm$.

 Let us represent the even real-valued form  \eqref{4.6} as
\begin{equation*}
  A=A_{+}+A_{-},
 \end{equation*}
 where
 \begin{equation}\label{4.11}
   A_{+}=\alpha^{0}x+\alpha^{12}e_{12}+\alpha^{13}e_{13}+\alpha^{23}e_{23},
  \end{equation}
 \begin{equation}\label{4.12}
   A_{-}=\alpha^{01}e_{01}+\alpha^{02}e_{02}+\alpha^{03}e_{03}+\alpha^4e.
  \end{equation}
  It is easy to check that $A_{+}$ commutes with $e_0$  and  $A_{-}$  anticommutes with it, i.e.
  \begin{equation}\label{4.13}
  e_0A_{\pm}=\pm A_{\pm}e_0.
 \end{equation}
 \begin{lem}
  The form $e_{0\mu}A_{-}$ commutes with $e_0$ and has the view    \eqref{4.11},  while
  $e_{0\mu}A_{+}$ anticommutes with $e_0$ and has the view    \eqref{4.12}  for any $\mu=1,2,3$.
  \end{lem}
  \begin{proof} For $\mu=1$ we have
  \begin{eqnarray*}
   e_{01}A_{-}=e_{01}(\alpha^{01}e_{01}+\alpha^{02}e_{02}+\alpha^{03}e_{03}+\alpha^4e)\\
   =\alpha^{01}x-\alpha^{02}e_{12}-\alpha^{03}e_{13}+\alpha^4e_{23}
  \end{eqnarray*}
  and
  \begin{eqnarray*}
   e_{01}A_{+}=e_{01}(\alpha^{0}x+\alpha^{12}e_{12}+\alpha^{13}e_{13}+\alpha^{23}e_{23})\\
   =\alpha^{0}e_{01}-\alpha^{12}e_{02}-\alpha^{13}e_{03}+\alpha^{23}e.
  \end{eqnarray*}
  The same proof remains valid for all other cases.
  \end{proof}
  \begin{thm}
 The forms  $\Omega^\pm=A\Psi^\pm$ are  non-trivial solutions of the discrete Hestenes equation if and only if the conditions
\begin{equation}\label{4.14}
 A_{\pm}=\frac{p_1e_{01}+p_2e_{02}+p_3e_{03}}{m-p_0}A_{\mp}
 \end{equation}
 hold,  or equivalently,
 \begin{equation}\label{4.15}
 A_{\mp}=-\frac{p_1e_{01}+p_2e_{02}+p_3e_{03}}{m+p_0}A_{\pm}.
 \end{equation}
\end{thm}
\begin{proof}
Let $A\Psi^-$ satisfy Eq.~\eqref{3.4}. Then  $A=A_{+}+A_{-}$ satisfies Eq.~\eqref{4.8}:
\begin{equation*}
   -\Big(\sum_{\mu=0}^3e_\mu p_\mu\Big)(A_{+}+A_{-})=m(A_{+}+A_{-})e_0.
 \end{equation*}
 From this we obtain
 \begin{equation*}
   -\Big(\sum_{\mu=1}^3e_0e_\mu p_\mu\Big)(A_{+}+A_{-})=p_0(A_{+}+A_{-})+me_0(A_{+}+A_{-})e_0.
 \end{equation*}
 Applying \eqref{4.13} we can rewrite the above relationship as
 \begin{equation*}
   -\Big(\sum_{\mu=1}^3e_0e_\mu p_\mu\Big)(A_{+}+A_{-})=(p_0+m)A_{+}+(p_0-m)A_{-}.
 \end{equation*}
 By Lemma~4.2, collecting like terms gives
 \begin{equation}\label{4.16}
   -(e_0e_1 p_1+e_0e_2 p_2+e_0e_3 p_3)A_{+}=(p_0-m)A_{-},
 \end{equation}
 \begin{equation*}
   -(e_0e_1 p_1+e_0e_2 p_2+e_0e_3 p_3)A_{-}=(p_0+m)A_{+}.
 \end{equation*}
 Conversely, substituting \eqref{4.14} into \eqref{4.15} yields the condition \eqref{4.10}. It follows that $A\Psi^-$ is  a  non-trivial solution of Eq.~\eqref{3.4}.

The same is true for  $A\Psi^+$ in place $A\Psi^-$.
 \end{proof}
It is clear that the conditions \eqref{4.14} and \eqref{4.15} can be rewritten as  systems of four linear algebraic equations. For example, from \eqref{4.16} we obtain
 \begin{eqnarray*}
(m-p_0)\alpha^{01}-p_1\alpha^{0}-p_2\alpha^{12}-p_3\alpha^{13}=0, \\
  (m-p_0)\alpha^{02}-p_2\alpha^{0}+p_1\alpha^{12}-p_3\alpha^{23}=0, \\
    (m-p_0)\alpha^{03}-p_3\alpha^{0}+p_1\alpha^{13}+p_2\alpha^{23}=0, \\
 (m-p_0)\alpha^4-p_1\alpha^{23}+p_1\alpha^{13}-p_3\alpha^{12}=0.
 \end{eqnarray*}
Thus for given $p_\mu$, $\mu=1,2,3$,  there are four linearly independent solutions of the form \eqref{4.5} for each positive and  negative  $p_0$.

\subsection*{Acknowledgment}
The author would like to thank N. Faustino for valuable discussions.

% ------------------------------------------------------------------------

\begin{thebibliography}{1}

\bibitem{B}  W.~E. Baylis ed., \textit{Clifford (Geometric) Algebra with Applications to
Physics, Mathematics, and Engineering}.  Birkh\"{a}user,  1996.

\bibitem{B1} W.~E. Baylis, \textit{Comment on `Dirac theory in spacetime algebra'.}  J. Phys. A: Math. Gen.   \textbf{35} (2002),  4791--4796.

\bibitem{Dezin} A.~A. Dezin, \textit{Multidimensional Analysis and Discrete Models.} CRC Press, Boca Raton, 1995.

\bibitem{FKS} N. Faustino, U. K\"{a}hler,  F. Sommen, \textit{Discrete Dirac operators in Clifford Analysis.}
Adv. Appl. Cliff. Alg. \textbf{17}(3) (2007),  451--467.

\bibitem{F2} N. Faustino,  \textit{Solutions for the Klein-Gordon and Dirac Equations on the Lattice Based on Chebyshev Polynomials}.
Complex Anal. Oper. Theory \textbf{10}(2) (2016), 379--399.

\bibitem{F3} N. Faustino,  \textit{A conformal group approach to the Dirac-K\"{a}hler system on the lattice.}
   Math. Methods Appl. Sci.  \textbf{40}(11) (2017),  4118--4127.

\bibitem{H1} D. Hestenes,  \textit{Real Spinor Fields.}  Journal of Mathematical Physics
 \textbf{8}(4) (1967), 798–-808.

\bibitem{H2} D. Hestenes, \textit{Spacetime Algebra.} Gordon and Breach,  New York, 1966.





\bibitem{Kahler}
 E. K\"{a}hler,   \textit{Der innere differentialk\"{u}l.}  Rendiconti di Matematica \textbf{21}(3--4) (1962), 425-523.

\bibitem{Kanamori}
I. Kanamori,  N. Kawamoto,  \textit{Dirac-K\"{a}hler fermion from Clifford product with noncommutative differential form on a lattice.}
  Int. J. Mod. Phys. A \textbf{19}(5) (2004),  695--736.

\bibitem{Rabin}
 J. M. Rabin, \textit{Homology theory of lattice fermion doubling.}  Nucl. Phys. B \textbf{201}(2) (1982),  315--332.

\bibitem{S1} V. Sushch,  \textit{A discrete model of the Dirac-K\"{a}hler equation.} Rep. Math. Phys. \textbf{73}(1) (2014),  109--125.
\bibitem{S2} V. Sushch, \textit{On the chirality of a discrete Dirac-K\"{a}hler equation.}  Rep. Math. Phys. \textbf{76}(2) (2015),  179--196.
\bibitem{S3} V. Sushch,  \textit{Discrete Dirac-K\"{a}hler equation and its formulation in algebraic form.} Pliska Stud. Math. \textbf{26} (2016), 225--238.
\bibitem{S4} V. Sushch,  \textit{Discrete Dirac-K\"{a}hler and Hestenes equations.} In Differential and difference
equations with applications. ICDDEA, Amadora, Portugal, May 18-22, 2015. Selected
contributions.	Springer Proceedings in Mathematics \& Statistics \textbf{164}, 433--442. Cham: Springer, 2016.

  \bibitem{Vaz} J. Vaz,  \textit{Clifford-like Calculus over Lattice.} Adv. Appl. Clifford Alg. \textbf{7}(1)  (1997), 37--70.

\end{thebibliography}
\end{document}